\documentclass[preprint]{elsarticle}

\usepackage{array,xspace,multirow,hhline,tikz,colortbl,tabularx,booktabs,fixltx2e,amsmath,amssymb,amsfonts,amsthm}
\usepackage{algorithm}
\usepackage{algorithmic}
\usepackage{verbatim,ifthen}
\usepackage{enumitem}
\usepackage{pifont}
\usepackage{ifthen}
\usepackage{calrsfs,mathrsfs}
\usepackage{bbding,pifont}
\usepackage{pgflibraryshapes}
\usepackage{nicefrac}

\definecolor{light-gray}{gray}{0.9}

	\newcommand{\eg}{e.g.,\xspace}

	\newtheorem{lemma}{Lemma}%
	\newtheorem{example}{Example}
     \newtheorem{fact}{Fact}%
    
    \newtheorem*{theorem*}{Theorem}

		\newcommand{\ml}[1][]{\ifthenelse{\equal{#1}{}}{\mathit{ML}}{\mathit{ML}(#1)}}
		\newcommand{\sml}[1][]{\ifthenelse{\equal{#1}{}}{\mathit{SML}}{\mathit{SML}(#1)}}
		\newcommand{\sd}[1][]{\ifthenelse{\equal{#1}{}}{\mathit{SD}}{\mathit{SD}(#1)}}
		\newcommand{\rsd}[1][]{\ifthenelse{\equal{#1}{}}{\mathit{RSD}}{\mathit{RSD}(#1)}}
		\newcommand{\st}[1][]{\ifthenelse{\equal{#1}{}}{\mathit{ST}}{\mathit{ST}(#1)}}
		\newcommand{\bd}[1][]{\ifthenelse{\equal{#1}{}}{\mathit{BD}}{\mathit{BD}(#1)}}
		\newcommand{\pc}[1][]{\ifthenelse{\equal{#1}{}}{\mathit{PC}}{\mathit{PC}(#1)}}
		\newcommand{\dl}[1][]{\ifthenelse{\equal{#1}{}}{\mathit{DL}}{\mathit{DL}(#1)}}
		\newcommand{\ul}[1][]{\ifthenelse{\equal{#1}{}}{\mathit{UL}}{\mathit{UL}(#1)}}

		\newcommand{\serdict}[1][]{\ifthenelse{\equal{#1}{}}{\sigma}{\sigma(#1)}}

	\newcommand\eat[1]{}

	\usepackage{enumitem}
	\setenumerate[1]{label=\rm(\it{\roman{*}}\rm),ref=({\it\roman{*}}),leftmargin=*}
	\newlength{\wordlength}

	\newcommand{\set}[1]{\{#1\}}
	\newcommand{\midd}{\mathbin{:}}

	\newcommand{\eqclass}[2][]{\ifthenelse{\equal{#1}{}}{[#2]}{[#2]_{\sim_{#1}}}}

	\newcommand{\Pref}[1][]{
		\ifthenelse{\equal{#1}{}}{\mathrel R}{\mathop{R_{#1}}}
	}                                          
	\newcommand{\sPref}[1][]{                  
		\ifthenelse{\equal{#1}{}}{\mathrel P}{\mathop{P_{#1}}}
	}                                          
	\newcommand{\Indiff}[1][]{                 
		\ifthenelse{\equal{#1}{}}{\mathrel I}{\mathop{I_{#1}}}
	}
	\newcommand{\prefset}[1][]{\ifthenelse{\equal{#1}{}}{\mathcal{R}}{\mathcal{R}_{#1}}}

\begin{document}

\title{A Comment on the Averseness of Random Serial Dictatorship to Stochastic Dominance Efficiency}

	 \author{Haris Aziz} \ead{haris.aziz@data61.csiro.au}

	\address{Data61 and UNSW} 
	


	\begin{abstract}
        Random Serial Dictatorship (RSD) is arguably the most well-known and widely used assignment rule. Although it returns an ex post efficient assignment, Bogomolnaia and Moulin (A new solution to the random assignment problem, J. Econ. Theory 100, 295--328) proved that RSD may not be SD-efficient (efficient with respect stochastic dominance). The result raises the following question: under what conditions is RSD not SD-efficient? In this comment, we give a detailed argument   that the RSD assignment is not SD-efficient if and only if an ex post assignment exists that is not SD-efficient. Hence RSD can be viewed as being inherently averse to SD-efficiency. The characterization was proved by Manea (2009).
\end{abstract}


\maketitle

\section{Introduction}

Consider the assignment problem in which $n$ agents expresses linear orders over $n$ objects and each agent is to be allocated one object~\citep{AbSo99a,ACMM05a,AMXY15a,BoMo01a,Gard73b, Sven94a,Sven99a}. 
 The most famous mechanism for the problem is \emph{random serial dictatorship (RSD)}~\citep{ABB13b, BoMo01a}. In RSD,  a permutation of the agents is chosen uniformly at random and then agents in the permutation are given the most preferred object that is still not allocated.
The reason RSD is a compelling mechanism for the assignment  problem is because it is strategyproof  and also ex post efficient (the outcome can be represented as convex combination of deterministic Pareto optimal outcomes). In fact, it has even been conjectured that RSD is the only mechanism that satisfies anonymity, strategyproofness and ex post efficiency~\citep[see \eg][]{LeSe11a}.

Although RSD is a desirable mechanism, \citet{BoMo01a} showed that RSD is not SD-efficient (efficiency with respect to stochastic dominance).\footnote{\citet{BoMo01a} used the term ordinal efficiency for SD-efficient. \citet{BoMo01a} also presented the probabilistic serial mechanism that is SD-efficient. However, the mechanism is not stratefyproof. } 
The observation was surprising because SD-efficiency is a very undemanding property. 
To highlight this point, note that if an assignment is not SD-efficient, then there exists another assignment in which for all cardinal utilities consistent with the ordinal preferences, all agents get at least as much utility and one agent gets strictly more utility. 
In this note, we explore the lack of SD-efficiency of RSD and present an elemetary and detailed argument for the following theorem.


    \begin{theorem*}[Proposition 3 of \citet{Mane09a}]
For a given preference profile, the RSD assignment is not SD-efficient if and only if there exists an ex post assignment that is not SD-efficient.
        \end{theorem*}

        The theorem highlights the fact not only is RSD SD-\emph{inefficient} in general but is never SD--efficient if there some ex post efficient assignment is not SD-efficient. 
        

%

\section{Preliminaries}

The model we consider is the \emph{random assignment problem}~\citep{BoMo01a} which is a triple $(N,O,\succ)$ where $N$ is the set of $n$ agents $\{1,\ldots, n\}$, $O=\{o_1,\ldots, o_n\}$ is the set of objects, and $\succ=(\succ_1,\ldots,\succ_n)$ specifies complete, anti-symmetric and transitive preferences $\succ_i$ of agent $i$ over $O$. 
We will denote by $\mathcal{R}(O)$ as the set of all complete and transitive relations over the set of objects $O$.

A random assignment $p$ is a $n\times n$ matrix $[p(i)(o_j)]_{1\leq i\leq n, 1\leq j\leq n}$ such that for all $i\in N$, and $o_j\in O$, $ p(i)(o_j) \in [0,1]$; $\sum_{i\in N}p(i)(o_j)= 1$ for all $j\in \{1,\ldots, n\}$; and $\sum_{o_j\in O}p(i)(o_j)= 1$ for all $i\in N$. 
The value $p(i)(o_j)$ represents the probability of object $o_j$ being allocated to  agent $i$. Each row $p(i)=(p(i)(o_1),\ldots, p(i)(o_n))$ represents the allocation of agent $i$. 
The set of columns correspond to probability vectors of the objects $o_1,\ldots, o_n$.
A feasible random assignment is \emph{discrete} if $p(i)(o)\in \{0,1\}$ for all $i\in N$ and $o\in O$. A discrete Pareto optimal assignment $p$ is \emph{Pareto optimal} if there does not exists another Pareto optimal assignment $q$ such that each agent gets the same or more preferred object in $q$ and at least one agent gets a more preferred object in $q$.

In order to reason about preferences over random allocations, we extend preferences over objects to preferences over random allocations. One standard extension is \emph{SD (stochastic dominance)} 
Given two random assignments $p$ and $q$, it holds that $p(i) \succsim_i^{\sd} q(i)$ i.e.,  a player $i$ \emph{$\sd$~prefers} allocation $p(i)$ to allocation $q(i)$ if for all $o\in O$.
	\[	\sum_{o_j\in \set{o_k\midd o_k\succsim_i o}}p(i)(o_j) \ge \sum_{o_j\in \set{o_k\midd o_k\succsim_i o}}q{(i)(o_j)}.\] 

	An assignment $p$ is \emph{$\sd$-efficient} is there exists no assignment $q$ such that $q(i) \succsim_i^{\sd} p(i)$ for all $i\in N$ and $q(i) \succ_i^{\sd} p(i)$ for some $i\in N$. An assignment is \emph{ex post efficient} if it can be represented as a probability distribution over the set of $\sd$-efficient discrete assignments.

Serial dictatorship (also called priority) is defined as follows. 
Each agent in permutation $\pi$ of $N$ gets a turn according to the permutation. When an agent's turn comes, the agent is allocated the most preferred object that is not yet allocated. The outcome of serial dictatorship is a deterministic assignment. We will refer to the outcome of serial dictatorship with respect to permutation $\pi$ as 
$\mathrm{Prio}(N,A,\succ,\pi)$.

RSD is a random assignment rule in which a permutation of agents is chosen uniformly at random and then serial dictatorship is run with respect to the permutation~\citep{AbSo98a,BoMo01a,ABB13b}:

\begin{equation*}
	\textit{RSD}(N,O,\succ)=\sum_{\pi\in \Pi^N}
	\frac{1}{n!}(\mathrm{Prio}(N,A,\succ,\pi))
\end{equation*}
where $\Pi^N$ denote the
set of all permutations of $N$.

 We say a random assignment $p$ is a \emph{proper} convex combination of a set of assignments $B$ if $p$ can be represented as a convex combination of assignments in $B$ such that the weight of each assignment in $B$ is non-zero. Note that the RSD assignment is a proper convex combination of the serial dictatorship outcomes.

 \begin{example}
 	Consider an assignment problem in which $N=\{1,2,3,4\}$, $O=\{o_1,o_2,o_3,o_4\}$ and the preferences $\succ$ are as follows.\footnote{The profile in this example is the same one that was used by \citet{BoMo01a} to show that the RSD outcome is not SD-efficient.}
 		\begin{align*}
 	1:&\quad o_1,o_2,o_3,o_4&
    3:&\quad o_2,o_1,o_4,o_3\\
 	2:&\quad o_1,o_2,o_3,o_4&
  	4:&\quad o_2,o_1,o_4,o_3
 	\end{align*}
 \[\mathrm{Prio}(N,O,\succ,1234)=\begin{pmatrix}
    	1&0&0&0\\
     	0&1&0&0\\
    	0&0&0&1\\
     	0&0&1&0\\
    	\end{pmatrix}, ~~ RSD(N,O,\succ)=\begin{pmatrix}
 	\nicefrac{5}{12}&\nicefrac{1}{12}&\nicefrac{5}{12}&\nicefrac{1}{12}\\
 \nicefrac{5}{12}&\nicefrac{1}{12}&\nicefrac{5}{12}&\nicefrac{1}{12}\\
 \nicefrac{1}{12}&\nicefrac{5}{12}&\nicefrac{1}{12}&\nicefrac{5}{12}\\
 \nicefrac{1}{12}&\nicefrac{5}{12}&\nicefrac{1}{12}&\nicefrac{5}{12}
 	\end{pmatrix}.\]

 In the RSD assignment, the probability of agent $1$ getting $o_1$ is $5/12$.
 \end{example}

         %
         %
         %
         %

\section{Inefficiency of RSD}

    Before we proceed, we present a characterization of SD-efficiency~ \citep[Lemma 3, ][]{BoMo01a}. An assignment $p$ admits a \emph{trading cycle} $o_0,i_0,o_1,i_1,\ldots, o_{k-1},i_{k-1},o_0$ in which $p(i_j)(o_j)>0$ for all $j\in \{0,\ldots, k-1\}$, $o_{j+1 \mod k} \succ_j o_{j\mod k}$ for all $j\in \{0,\ldots, k-1\}$.
 \citet{BoMo01a} proved that an assignment is SD-efficient if and only if it does not admit a trading cycle.

 \begin{fact}[\citet{BoMo01a}]\label{fact:sdeff}
     An assignment is SD-efficient if and only if  it does not admit a trading cycle.
     \end{fact}

We will also use the following characterization of Pareto optimal discrete assignments~\citep{AbSo98a}. Fact~\ref{fact:AbSo} also follows from Proposition 1 by \citet{BrKi05a}.

 \begin{fact}[\citet{AbSo98a}]\label{fact:AbSo}
     A discrete assignment is Pareto optimal if and only if it is an outcome of serial dictatorship.
     \end{fact}

     We first present a simple lemma. 
     
     \begin{lemma}\label{lemma:convex-non-zero}
         Consider any assignment $p$ that is a proper convex combination of all discrete Pareto optimal assignments. Then $p(i)(o)>0$ if there exists some Pareto optimal discrete assignment $q$ such that $q(i)(o)>0$.
         \end{lemma}
         \begin{proof}
             Assume that there exists some Pareto optimal discrete assignment $q$ such that $q(i)(o)>0$.  Since $q$ has non-zero probability when $r$ is expressed as a convex combination of discrete Pareto optimal assignments, it follows that $q(i)(o)>0$.
             \end{proof}
     
     Next, we rely on Lemma~\ref{lemma:convex-non-zero} and Fact~\ref{fact:sdeff} to prove the following. 

\begin{lemma}\label{lemma:convex-expost}
    For a given preference profile, if there exists a random assignment that is ex post efficient but not SD-efficient, then a proper convex combination of all Pareto optimal deterministic assignments is such a random assignment as well.
    \end{lemma}
    \begin{proof}
        Assume that there  exists some assignment $s$ that is ex post but not SD efficient. Consider any assignment $p$ that is a proper convex combination of all discrete Pareto optimal assignments. We will show that $p$ is not SD-efficient. 
        
        Since $s$ is ex post efficient, it can be represented by a convex combination of Pareto optimal discrete assignments. Let the set be $B$. Since $s$ is not SD-efficient, by Fact~\ref{fact:sdeff}, it admits a trading cycle $o_0,i_0,o_1,i_1,\ldots, o_{k-1},i_{k-1},o_0$ in which $s(i_j)(o_j)>0$ for all $j\in \{0,\ldots, k-1\}$, $o_{j+1 \mod k} \succ_j o_{j\mod k}$ for all $j\in \{0,\ldots, k-1\}$. Since $s$ is ex post efficient, it can be represented as a convex combination of Pareto optimal assignments. Therefore if 
        $s(i_j)(o_j)>0$, then there exist some discrete Pareto optimal assignment $r$ such that $r(i_j)(o_j)>0$. As $p$ is proper convex combination of all discrete Pareto optimal assignments, if there exists a discrete assignment $r$ such that $r(i_j)(o_j)>0$, then by Lemma~\ref{lemma:convex-non-zero}, $p(i_j)(o_j)>0$ for all $j\in \{0,\ldots, k-1\}$.
By this argument, it follows that $p$ admits a trading cycle  $o_0,i_0,o_1,i_1,\ldots, o_{k-1},i_{k-1},o_0$ in which $p(i_j)(o_j)>0$ for all $j\in \{0,\ldots, k-1\}$, $o_{j+1 \mod k} \succ_j o_{j\mod k}$ for all $j\in \{0,\ldots, k-1\}$.  Hence $p$ is not SD-efficient.   
        \end{proof}
        
        Next, we use  Fact~\ref{fact:AbSo} to prove the following. 
        
        \begin{lemma}\label{lemma:rsd-convex}
            The outcome of RSD is a proper convex combination of all Pareto optimal deterministic assignments.
            \end{lemma}
            \begin{proof}
                RSD can be viewed as applying serial dictatorship with respect to all the $n!$ permutations and then aggregating each of the outcomes weighted with the probability $1/n!$ By Fact~\ref{fact:AbSo}, each discrete Pareto optimal assignment is a result of serial dictatorship with respect to some permutation.  Hence the RSD assignment is a proper convex combination of all Pareto optimal deterministic assignments.
                \end{proof}
                
               By using Lemmas~\ref{lemma:convex-expost} and \ref{lemma:rsd-convex}, we can get the following statement. 
                
                \begin{theorem*}[\citet{Mane09a}]
For a given preference profile, the RSD assignment is not SD-efficient if and only if there exists an ex post assignment that is not SD-efficient.
                    \end{theorem*}
                    \begin{proof}
                        First assume that for the given preference profile, every ex post assignment is SD-efficient. Since the RSD assignment is ex post efficient, it follows that the RSD assignment is SD-efficient. 
                        
                        We now assume that there exists an ex post assignment that is not SD-efficient (1).
                        By Lemma~\ref{lemma:rsd-convex}, the RSD assignment is a proper convex combination of all Pareto optimal deterministic assignments (2). By (1) and (2), it follows from Lemma~\ref{lemma:convex-expost}, that the RSD assignment is not SD-efficient.
                        \end{proof}

		\paragraph{Acknowledgments}
	Data61 is funded by the Australian Government through the Department of Communications and the Australian Research Council through the ICT Centre of Excellence Program. Thanks to Mark Wilson for pointing out that the characterization was already present in a proposition in the paper by \citet{Mane09a}.

\renewcommand{\bibfont}{\normalfont\small}

\end{document}